\documentclass[journal]{IEEEtran}

\usepackage{fancyhdr}
\usepackage{latexsym}
\usepackage{graphicx,color}
\usepackage{shadow,epsf,amsthm,amssymb,amsmath}
\usepackage{enumerate}
\usepackage{setspace}                           
\usepackage[all,2cell]{xy} \UseAllTwocells \SilentMatrices
\usepackage{soul} 
\usepackage{epstopdf} 

\usepackage{tikz}	
\usetikzlibrary{backgrounds,fit,decorations.pathreplacing}  
\usetikzlibrary{automata} 
\usetikzlibrary{chains}

\usepackage{lipsum}
\usepackage[colorinlistoftodos]{todonotes}

\newtheorem{definitionenv}{Definition}
\newtheorem{lemmaenv}{Lemma}
\newtheorem{theoremenv}[lemmaenv]{Theorem}
\newtheorem{corollaryenv}[lemmaenv]{Corollary}
\newtheorem{propositionenv}[lemmaenv]{Proposition}
\newtheorem{conjectureenv}[lemmaenv]{Conjecture}
\newtheorem{exampleenv}{Example}
\newtheorem{app-lemmaenv}[section]{Lemma}

\newenvironment{definition}{\begin{definitionenv}\rm}{\end{definitionenv}}
\newenvironment{lemma}{\begin{lemmaenv}\rm}{\end{lemmaenv}}
\newenvironment{theorem}{\begin{theoremenv}\rm}{\end{theoremenv}}
\newenvironment{corollary}{\begin{corollaryenv}\rm}{\end{corollaryenv}}
\newenvironment{example}{\begin{exampleenv}\rm}{\end{exampleenv}}
\newenvironment{proposition}{\begin{propositionenv}\rm}{\end{propositionenv}}
\newenvironment{conjecture}{\begin{conjectureenv}\rm}{\end{conjectureenv}}
\newenvironment{app-lemma}{\begin{app-lemmaenv}\rm}{\end{app-lemmaenv}}

\newcommand{\bd}{\begin{definition}}
\newcommand{\ed}{\end{definition}}
\newcommand{\bl}{\begin{lemma}}
\newcommand{\el}{\end{lemma}}
\newcommand{\elp}{\hspace*{\fill} $\Box$
                 \end{lemma}}
\newcommand{\bt}{\begin{theorem}}
\newcommand{\et}{\end{theorem}}
\newcommand{\etp}{\hspace*{\fill} $\Box$
                 \end{theorem}}
\newcommand{\bc}{\begin{corollary}}
\newcommand{\ec}{\end{corollary}}
\newcommand{\ecp}{\hspace*{\fill} $\Box$
                 \end{corollary}}
\newcommand{\bcj}{\begin{conjecture}}
\newcommand{\ecj}{\end{conjecture}}

\newcommand{\be}{\begin{example}}
\newcommand{\ee}{\end{example}}
\newcommand{\eep}{\hspace*{\fill} $\Box$
                 \end{example}}
\newcommand{\bp}{\begin{proposition}}
\newcommand{\ep}{\end{proposition}}
\newcommand{\epp}{
                 \end{proposition}}

\newcommand{\bra}[1]{\langle#1|}
\newcommand{\ket}[1]{|#1\rangle}
\newcommand{\braket}[2]{\langle#1|#2\rangle}

\newcommand{\wt}[1]{\text{wt}\left(#1\right)}
\newcommand{\diag}{\text{diag}}

\newcommand{\eeq}{ \setcounter{equation} {\value{enumi}}}

\newcommand{\cC}{\mathcal{C}}

\newcommand{\cF}{\mathcal{F}}
\newcommand{\cG}{\mathcal{G}}
\newcommand{\cH}{\mathcal{H}}

\newcommand{\cS}{\mathcal{S}}

\newcommand{\cV}{\mathcal{V}}
\newcommand{\cW}{\mathcal{W}}

\newcommand{\mF}{\mathbb{F}}
\newcommand{\mC}{\mathbb{C}}

\newcommand{\bfa}{{\mathbf a}}
\newcommand{\bfb}{{\mathbf b}}
\newcommand{\bfc}{{\mathbf c}}

\newcommand{\bfu}{{\mathbf u}}
\newcommand{\bfv}{{\mathbf v}}

\newcommand{\bfp}{{\mathbf p}}

\newcommand{\bfw}{{\mathbf w}}

\def\F{{\mathbb F}}

\newcommand{\tm}[1]{\textnormal{#1}}
\def\x{\times}

\def\beq{\begin{equation}}
\def\eeq{\end{equation}}

\def\bean{\begin{IEEEeqnarray*}{rCl}}
\def\eean{\end{IEEEeqnarray*}}

\begin{document}

\date{}
\title{\Large\bf On the MacWilliams Identity for Classical and Quantum Convolutional Codes}

\author{Ching-Yi Lai, Min-Hsiu Hsieh, and Hsiao-feng Lu\thanks{
Part of this work was in Proceedings of IEEE Intl. Symp. Inf. Theory 2014 \cite{LH14a},
and part of this work was  in Proceedings of IEEE Information Theory Workshop 2014 \cite{LHL14b}.

C.-Y. Lai and M.-H. Hsieh are with the Centre for Quantum Computation \& Intelligent Systems,
Faculty of Engineering and Information Technology,
University of Technology, Sydney, New South Wales, Australia 2007. M.-H. Hsieh is also with the UTS-AMSS Joint Research Laboratory for Quantum Computation and Quantum Information Processing, Academy of Mathematics and Systems Science, Chinese Academy of Sciences, Beijing 100190, China.
(emails: cylai0616@gmail.com and Min-Hsiu.Hsieh@uts.edu.au)

H.-F. Lu is with the Department of Electrical and Computer Engineering, at National Chiao Tung University,  Taiwan.
(email: francis@mail.nctu.edu.tw).
}
}

\maketitle

\thispagestyle{empty}

\begin{abstract}
The weight generating functions associated with convolutional codes (CCs) are based on state space realizations or the weight adjacency matrices (WAMs).
The MacWilliams identity for CCs on the WAMs was first established by Gluesing-Luerssen and Schneider in the case of minimal encoders, and generalized by Forney.
We consider this problem in the viewpoint of constraint codes and obtain a simple and direct proof of this MacWilliams identity in the case of minimal encoders.
For our purpose, we choose a different representation for the exact weight generating function (EWGF) of a  block code, by
defining it as a linear combination of orthonormal vectors in Dirac bra-ket notation. 
This representation provides great flexibility so that general split weight generating functions and their MacWilliams identities
can be easily obtained from the MacWilliams identity  for  EWGFs. 
As a result,
we also obtain the MacWilliams identity for the input-parity weight adjacency matrices  of a  systematic convolutional code and its dual.
Finally, paralleling the development of the classical case, we establish the MacWilliams identity for quantum convolutional codes.

\end{abstract}

\section{Introduction}

In coding theory,  a fundamental theorem is the MacWilliams identity for linear block codes,
which provides a precise relation between the weight generating functions of a code and its dual \cite{MS77}.
The weight generating function of a code details the distribution of codeword weights, which can be used to analyze the error performance of the code.

Convolutional codes (CCs)  offer a rather different coding paradigm {from block codes}. The convolutional structure allows a much lower complexity for encoding and decoding circuits without deteriorating its error-correcting ability \cite{Forney70}.
The \emph{free-distance enumerator}  is the first  notion that counts the weight distribution of the fundamental paths  that start and end in the zero states of a code's state diagram without passing any intermediate zero states  \cite{Viterbi71}. This free-distance enumerator is crucial in the error analysis of a CC; however, it was realized later that the MacWilliams identity does not hold for the free distance enumerators \cite{SM77}.

A more refined notion of a weight generating function is the weight adjacency matrix (WAM) \cite{MS67,AG92,McE02}. Each entry of this matrix is the weight distribution of all outputs associated with the corresponding state transitions. Unfortunately, a general WAM strongly depends on the underlining encoder and state space description.
However, the WAM is shown to be an invariant of a CC if the encoder is minimal \cite{GL05}. A breakthrough was made by Gluesing-Luerssen and Schneider in~\cite{GLS08,GLS09}, where the MacWilliams identity for the WAMs of a CC and its dual is established.
Later, Forney employed the normal graph duality theorem \cite{Forney01} and obtained a more general MacWilliams theorem for various weight generating functions of CCs without the requirement of minimal encoder. 

This paper outlines a direct proof of the MacWilliams identity for CCs with minimal encoders. First we  define the exact weight generating function (EWGF) of a linear block code as a linear combination of orthonormal vectors in Dirac bra-ket notation. The EWGFs of a code and its dual are directly related by a Fourier transform operator, which gives the MacWilliams identity for linear block codes~\cite{MS77} and facilitates the derivation of   MacWilliams identities for general split weight generating functions.
As for CCs, we begin with the \emph{constraint codes} of a CC, introduced by Forney \cite{Forney11}, which detail the state evolutions in the state diagram.
We show  that the dual of a CC with minimal encoder can be defined by the dual of its constraint codes.
As a result, we obtain a direct proof of the MacWilliams identity for CCs with minimal encoders, which is simpler than that in~\cite{GLS09}, without using Forney's normal graph duality theorem. 
(The assumption of minimal encoders can be removed by using the normal graph duality theorem \cite{Forney11,Forney01}.)
In addition, the scalar in the MacWilliams identity for CCs that is missing in \cite{Forney11} is explicitly given in our proof.
Moreover, 
we derive a MacWilliams identity for the input-parity WAMs (IPWAMs) of CCs with systematic encoders, and this answers an open question proposed by Gluesing-Luerssen and Schneider~\cite{GLS09}. 
Finally, we consider  weight enumeration over all codewords of a  CC and prove its MacWilliams identity.
Following that, relations between various weight enumerations, including the free-distance enumerators, are discussed. 

Quantum convolutional codes (QCCs) receive great attention for their capabilities in protecting a stream of quantum information in quantum communication, since large blocks of quantum information are very fragile to decoherence \cite{OT03,FGG07}.
The WAMs and free-distance enumerators of QCCs are defined accordingly and they function like the classical counterparts~\cite{PTO09,WHZ14}. 
We proceed to define the dual code of a QCC within the framework of entanglement-assisted quantum  convolutional codes (EA-QCCs) \cite{WB09,WB10}.
An EA-QCC is defined by a constraint code, which is an EA stabilizer code~\cite{BDM06,BDM062}.
The dual code of a stabilizer code is an EA stabilizer code~\cite{LBW13} and their  MacWilliams identity exist~\cite{SL96,Rains96b,AL99}.
Our notion of duality follows the normal factor graph duality theorem 
\cite{Forney01,Forney11,MK05,ABM11}.
Then our classical treatment directly paves the way for the establishment of the MacWilliams identity for EA-QCCs.


This paper is organized as follows. We introduce the Dirac bra-ket notations and basics of linear block codes and the
MacWilliams identity in the next section, so that the materials in this paper are self-complement.
 In Sec. III, we first discuss
classical convolutional codes in the viewpoint of constraints codes and then prove the MacWilliams identities for various
notions of weight enumerations of CCs. 
The MacWilliams identity for EA-QCCs is given in Sec. IV, as well as the definition of the dual code of an EA-QCC. The
conclusion follows in Sec. V.

\section{The MacWilliams Identity for Split Weight Generating Functions in the Dirac Notation} \label{SecII_Pre}
We define notations, review the MacWilliams theorem for orthogonal groups, and introduce the Dirac bra-ket notation in this section.
This representation allows us to easily derive the MacWilliams identity for general split weight generating functions and facilitates the establishment of our results in the following sections. The readers would have a better understanding of the Dirac notation, which is used in both  classical and quantum cases throughout this article.

\subsection{The MacWilliams Identity and the Dirac Notation}

We begin with the Dirac notation, which is used throughout this article.
Let $\ket{\phi}$ denote a vector in a complex Hilbert space $\mathcal{H}$, and its adjoint is denoted by $\bra{\phi}=\ket{\phi}^{\dag}$.
For $\ket{\phi_1},\ket{\phi_2}\in \mathcal{H}$, their inner and  outer products are denoted by  $\braket{\phi_1}{\phi_2}$ and $\ket{\phi_1}\bra{\phi_2}$, respectively.

Suppose $\cV$ is a vector space of dimension $n$ over a finite field $\mathbb{F}_q$, where $q=p^m$ is a prime power, i.e., $\cV\equiv \mathbb{F}_q^n$.
We define a complex Hilbert space  $\mathcal{H}_\mathcal{V}$, corresponding to $\cV$, with an orthonormal basis $$\{\ket{{  \bfv}}: \bfv\in \cV\}.$$
 This means that
$\braket{{  \bfv}}{{  \bfv}'}=\delta_{{  \bfv},{  \bfv}'},$ where $\delta_{{  \bfv},{  \bfv}'}$ is the Kronecker delta function.
Thus $\cH_\cV$ has dimension $q^n$.
Note that the bold-faced letter $\bfv$ is used to denote a vector in $\cV$, while $\ket{\bfv}$ is a vector in $\cH_{\cV}$.
We will define weight enumeration of a subset of $\cV$ as a linear combination of orthonormal vectors in $\cH_{\cV}$, rather than a multivariate polynomial commonly used in the literature (see, e.g., Ref.~\cite{MS77}).
\begin{definition}\label{def:EWGF}
The \emph{exact weight generating function} (EWGF) $g^E_{\mathcal{C}}$ of a set $\mathcal{C} \subset \mathcal{V}$ is defined as
\begin{align} g^E_{\mathcal{C}}=\sum_{\bfv\in \mathcal{C}}\ket{\bfv} \in\cH_{\cV}.\end{align}
\end{definition}

\be
An  $[n,k,d]_q$ linear block code over $\mathbb{F}_q$ is a $k$-dimensional  subspace of $\cV$.
Suppose $\cC$ is the $[7,3]_2$ simplex code with codewords
$0000000$,
$0110011$,
$1010101$,
$1100110$,
$0001111$,
$0111100$,
$1011010$,  and $1101001$ \cite{MS77}.
Then $g^E_{\mathcal{C}}=\ket{0000000}+\ket{0110011}+\ket{1010101}+\ket{1100110}+\ket{0001111}+\ket{0111100}+\ket{1011010}+\ket{1101001}$.
\eep


Let $\langle \ , \ \rangle_{\cV}: (\cV, \cV)\rightarrow \mathbb{F}_p$ denote an inner product in $\cV$. We may sometimes omit the subscript  when the underlying vector space is clear from the context.
The dual set of $\cC$ in $\cV$~is
\begin{align}
\cC^{\perp}=\{\bfv'\in \cV: \langle \bfv',\bfv \rangle_{\cV}=0, \forall \bfv\in \cC \}.
\end{align}
Let $\mathcal{F}_{\mathcal{V}}$ be a Fourier transform operator defined by
\begin{align} \label{eq:Fourier}
\mathcal{F}_{\mathcal{V}}= \sum_{v'\in{\mathcal{V}}}\sum_{v\in\mathcal{V}} \omega^{\langle v',v\rangle} \ket{v'}\bra{v}
\end{align}
and $\omega=e^{2\pi i/p}$ is a primitive complex $p$-th root of unity.
(Note that $\cF_{\cV}$ is not normalized: $\cF^{\dag}_{\cV}\cF_{\cV}=|\cV|\mathbb{I}$, where $|\cV|=q^n$ and $\mathbb{I}$ is the identity operator.)
The  MacWilliams identity connects the weight enumerators of an \emph{additively closed} set and its dual~\cite{MS77}.
It can be rephrased as a Fourier transform in the Dirac notation as in the following theorem.
This idea has been used in the construction of quantum Calderbank-Shor-Steane (CSS) codes \cite{CS96,Ste96},
where $g_\cC^E$, up to a normalization factor, is the logical zero state of the CSS code, defined by a classical linear dual-containing code $\cC$.
\bt\label{thm:MacWilliamsIdentity_EWGF}
Suppose $\mathcal{C}$ is an {additively closed} subset of $\mathcal{V}$ with an EWGF $g_{\mathcal{C}}^E=\sum_{v\in \mathcal{C}}\ket{v}$.
Then the EWGF of its dual set $\mathcal{C}^{\perp}$  is
\begin{align}
g^E_{\mathcal{C}^{\perp}} = \frac{1}{|\mathcal{C}|} \mathcal{F}_{\mathcal{V}} \ g_{\mathcal{C}}^E. \label{eq:MI for EWGF}
\end{align}
\et
The proof is straightforward (see~\cite{LH14a}). It is natural to generalize this result to the direct product of spaces.
Suppose $\cW$ is a vector space of dimension $c$ over $\mathbb{F}_{q'}$, where $q'=p^{m'}$. (Note that $m'$ may or may not equal to $m$.)
We can form a product vector space $\cV\times \cW$ with  an inner product 
\begin{align}
\langle (\bfv_1:\bfw_1), (\bfv_2:\bfw_2) \rangle_{\cV\times \cW}= \langle \bfv_1,\bfv_2 \rangle_{\cV} +\langle \bfw_1, \bfw_2 \rangle_\cW \label{eq:product_space_inner_product}
\end{align}
for $\bfv_1,\bfv_2\in \cV$ and $\bfw_1,\bfw_2\in \cW$,
where the addition is in $\mF_p$.
{We use the notation $(\bfa:\bfb)$ to denote the concatenation of two vectors $\bfa$ and $\bfb$.}
 Thus we can define a tensor product space $\cH_{\cV\times \cW}=\cH_\mathcal{V}\otimes \cH_\mathcal{W}$, which is spanned by $$\{\ket{\bfv}\otimes \ket{\bfw}: \bfv\in \mathcal{V},{\bfw}\in \mathcal{W}\}.$$
The tensor product $\otimes$
of  two matrices $A$  and   $B$  is defined as
\[
A\otimes B =
\begin{bmatrix}
A_{1,1} B & A_{1,2}B &\cdots &A_{1,t} B\\
A_{2,1} B & A_{2,2}B &\cdots &A_{2,t} B\\
\vdots&\vdots &\ddots&\vdots\\
A_{s,1} B & A_{s,2}B &\cdots &A_{s,t} B\\
\end{bmatrix},
\]
where $A=[A_{i,j}]$ is of dimensions $s\times t$.
Likewise, we can define a Fourier transform operator $\cF_{\cV\times \cW}$ on $\cH_\cV\otimes \cH_\cW$:
\[
\mathcal{F}_{\mathcal{V}\times \cW}= \sum_{(\bfv':\bfw')\in{\mathcal{V}\times \cW}\atop (\bfv:\bfw)\in\mathcal{V}\times \cW} \omega^{\langle (\bfv':\bfw'),(\bfv:\bfw)\rangle} \ket{\bfv'}\otimes \ket{\bfw'}\bra{\bfv}\otimes \bra{\bfw}.
\]
By replacing $\cV$ with $\cV\times \cW$ in Theorem \ref{thm:MacWilliamsIdentity_EWGF}, we can immediately obtain the following corollary.
\bc\label{cor:tensor MI_EWGF}
Suppose $\mathcal{C}$ is an additively closed subset of $\mathcal{V}\times
\cW$ with EWGF $g_{\mathcal{C}}^E=\sum_{v\in \mathcal{C}}\ket{v}$.
Then the EWGF of its dual set $\mathcal{C}^{\perp}$  is
\begin{align}
g^E_{\mathcal{C}^{\perp}} = \frac{1}{|\mathcal{C}|} \mathcal{F}_{\mathcal{V}\times\cW} \ g_{\mathcal{C}}^E. \label{eq:tensor MI for EWGF}
\end{align}
\ec
Recently a type of data-and-syndrome error-correcting stabilizer codes are introduced~\cite{ALB16}, where codes are defined over $\mathbb{F}_4^n\times \mathbb{F}_2^m$.
Thus Corollary \ref{cor:tensor MI_EWGF} can be applied to these codes.

The Fourier transform operator has a nice property that it can be decomposed as a tensor product of Fourier transform operators on its components as in the following lemma.
\bl \label{lemma:fourier transform}
\[
\mathcal{F}_{\mathcal{V}\times \cW}= \mathcal{F}_{\mathcal{V}} \otimes \mathcal{F}_{\mathcal{W}}.
\]
In particular, $\cF_{\mathbb{F}_q^n}= \cF_{\mathbb{F}_q}^{\otimes n}$.
\el

\subsection{General Split Weight Generating Functions}
In this {subsection} we first define the notion of \emph{split weight generating function}  defined by a linear functional.
A linear functional that maps $\cH_{\mF_q}$ to the complex numbers $\mC$ can be written as
$$
\sum_{\alpha\in \mF_q}  x_{\alpha}\bra{\alpha},
$$
where $x_{\alpha}\in \mC$.
\bd \label{def:PWE}
A \emph{split weight generating function} (SWGF) of   $\cC \subset \mathbb{F}_q^n $, defined by a linear functional  $\gamma: \cH_{\mathbb{F}_q^n}\rightarrow \mathbb{C}$ is
\begin{align}
g_{\cC}(\gamma)=\left(\bigotimes_{j=1}^n \gamma_j \right) g^{E}_{\cC},
\end{align}
where $\gamma = \bigotimes_{j=1}^n \gamma_j$ with $\gamma_j:\cH_{\mathbb{F}_q}\rightarrow \mathbb{C}$.
\ed

The MacWilliams identity holds for the SWGFs of an additively closed set $\cC$ and its dual, and can be easily obtained as an application of Theorem~\ref{thm:MacWilliamsIdentity_EWGF}.
\bt \label{thm:PWE_MI}
The MacWilliams identity  for the SWEs, defined by $\bigotimes_j \gamma_j$, of an additively closed set $\cC\subset \mathbb{F}_q^n$ and its dual is
\begin{align}
g_{\mathcal{C}^{\perp}}\left(\bigotimes_j \gamma_j\right) =& \frac{1}{|\mathcal{C}|} g_\mathcal{C}\left(\bigotimes_j \left(\gamma_j\cF_{\mF_q}\right)\right). \label{eq:MacWilliams Identity}
\end{align}
\et
\begin{proof}

\begin{align*}
&\frac{1}{|\cC|}g_\mathcal{C}\left(\bigotimes_j \left(\gamma_j\cF_{\mF_q}\right)\right)\stackrel{(a)}{=} \frac{1}{|\cC|}\bigotimes_j \left(\gamma_j\cF_{\mF_q}\right) g_{\cC}^{E}\\
\stackrel{(b)}{=}& \left(\bigotimes_j \gamma_j \right) \frac{1}{|\cC|}  \cF_{\mathbb{F}_q^n} g_{\cC}^E
\stackrel{(c)}{=} \left(\bigotimes_j \gamma_j \right)  g_{\mathcal{C}^{\perp}}^E
\stackrel{(d)}{=} g_{\mathcal{C}^{\perp}}\left(\bigotimes_j \gamma_j\right),
\end{align*}
where $(a)$, $(d)$ are by Definition \ref{def:PWE}; $(b)$  is from Lemma \ref{lemma:fourier transform}; and $(c)$
follows from Theorem \ref{thm:MacWilliamsIdentity_EWGF}.

\end{proof}

Various notions of weight enumeration can be defined by an appropriate $\gamma$,
and their corresponding MacWilliams identities are an direct application of Theorem~\ref{thm:PWE_MI}.
The Hamming weight  of a vector $\bfv\in\mathbb{F}_q^n$ is the number of its nonzero components and {is denoted by $\wt{\bfv}$}.
Two related weight generating functions are as follows.
\be
(Hamming weight generating function (HWGF))
Define a linear functional 
\begin{align}
\label{eq:gamma Hamming}
\gamma_\text{H}(y)= \bra{0}+\sum_{\alpha \in \mF_q\setminus \{0\}} y\bra{\alpha},
\end{align}
where $y\in \mC$.
The HWGF of a code $\cC \subset \mF_q^n$ is
\begin{align}
g_{\mathcal{C}}^H(y)\triangleq&\sum_{i=0}^n A_i y^i=  g_{\cC}(\gamma^{\otimes n}_{\text{H}}(y)), 
\end{align}
where $A_i$ is the number of codewords in $\cC$ of weight $i$.
By Theorem~\ref{thm:PWE_MI}, we have
\begin{align}
g_{\mathcal{C}^{\perp}}^H(y) =& \frac{(1+(q-1)y)^n}{|\mathcal{C}|} g_\mathcal{C}^H\left(\frac{1-y}{1+(q-1)y}\right). \label{eq:MacWilliams Identity}
\end{align}
\ee

\be \label{ex:ip}
(Input-parity weight generating function (IPWGF))
Consider an $[n,k]_q$ linear block code $\cC$ with a systematic generator matrix $G=\begin{pmatrix}I_k &A\end{pmatrix}$.
The first $k$ symbols of a codeword in $\cC$ are  information symbols, while the remaining  $n-k$ symbols are parity symbols.
We say a codeword $\bfc=(\bfc_I:\bfc_P)\in \cC$, where $\bfc_I\in \mF_q^k$ and $\bfc_P\in \mF_q^{n-k}$, has logical weight $i$, parity weight $r$, and output weight $o$ if $\wt{\bfc_I}=i$, $\wt{\bfc_P}=r$,  and $\wt{\bfc}=o$.
The IPWGF of $\cC$ is
\begin{align}
g_{\mathcal{C}}^{IP}(x,y) \triangleq & \sum_{i=0}^{k}\sum_{j=0}^{n-k} A_{i,j} x^i y^j= g_{\cC}\left( \gamma_H^{\otimes k}(x)\otimes \gamma_H^{\otimes n-k}(y)\right),  \label{eq:12}
\end{align}
where $A_{i,j}$ is the number of codewords in $\cC$ of logical weight $i$ and parity weight $j$.

%


The $[n,n-k]_q$ dual code $\cC^{\perp}$  has a systematic generator matrix $H=\begin{pmatrix}-A^T &I_{n-k} \end{pmatrix}$.
In this case, the last $n-k$ symbols of a codeword in $\cC^{\perp}$ are   information symbols. 
Thus we have to switch the indeterminates  $x$ and $y$ in (\ref{eq:12})  when applying Theorem~\ref{thm:PWE_MI}.
After some calculations, we have 
\begin{align}
g_{\mathcal{C}^{\perp}}^{IP}(x,y ) 
=&\frac{(1+(q-1)y)^{k}(1+(q-1)x)^{n-k}}{|\mathcal{C}|} \notag\\
 &\cdot g_\mathcal{C}^{IP}\left( \frac{1-y}{1+(q-1)y},\frac{1-x}{1+(q-1)x}\right). \label{eq:MacWilliams Identity IPWGF}
\end{align}
The input-output weight generating function (IOWGF) of $\cC$ can be similarly defined  as
\begin{align}
g_{\mathcal{C}}^{IO}(x,y) \triangleq & \sum_{i=0}^{n}\sum_{j=0}^{k} B_{i,j} x^i y^j = g_{\mathcal{C}}^{IP}(xy,y), \label{eq:14}
\end{align}
where $B_{i,j}$ is the number of vectors in $\cC$ of  output weight $i$ and logical weight $j$.
From (\ref{eq:MacWilliams Identity IPWGF}), we have
\begin{align}
g_{\mathcal{C}^{\perp}}^{IO}(x,y ) 
=&\frac{(1+(q-1)y)^{k}(1+(q-1)xy)^{n-k}}{|\mathcal{C}|} \notag\\
 &\cdot g_\mathcal{C}^{IP}\left( \frac{1-y}{1+(q-1)y},\frac{1-xy}{1+(q-1)xy}\right). \label{eq:MacWilliams Identity IOWGF}
\end{align}
{The IPWGF  has been introduced in the analysis of error performance~\cite{LKY04,CL09},
with the help of its MacWilliams identity. }
We will derive similar equations for convolutional codes in the next section,
which can be used in error performance analysis.

\ee


\section{Convolutional Codes and the MacWilliams Identity} \label{sec:MI for CC}


\subsection{Constraint Code, Weight Adjacency Matrix, and the MacWilliams Identity}
Let $\cal C$ be an $(n,k,m)$ convolutional code over $\F_q$ with polynomial generator matrix $G(D) \in {\bf M}_{k \x n}\left( \F_q[D]\right)$ for some indeterminate $D$ and
  overall constraint length $m$.  Then $\cal C$ is a rank-$k$ submodule of $\left( \F_q[D] \right)^n$ given by 
$$
{\cal C} = \left\{ \bfu(D)  G(D): \bfu(D) \in \left( \F_q[D] \right)^k \right\}.
$$
Assume further that $G(D)$ is a \emph{minimal} encoder of $\cC$
with $m$ memory symbols.
That is, $G(D)$ is a basic encoder with a right inverse and its overall constraint length $m$ is  equal to the maximum degree of its $k\times k$ subdeterminants \cite{Forney70}.

As shown in Fig.~\ref{fig:convolutional encoder}, $G(D)$ outputs $n$ symbols from $k$ logical input symbols at each time step.
\begin{figure}
\vspace{-0.5cm}
\centerline{
    \begin{tikzpicture}[scale=0.6][very thick]
    \fontsize{7pt}{1} 
    \tikzstyle{halfnode} = [draw,fill=white,shape= underline,minimum size=1.0em]
    \tikzstyle{checknode} = [draw,fill=blue!10,shape= rectangle,minimum height=4em, minimum width=2em]
    \tikzstyle{variablenode} = [draw,fill=white, shape=circle,minimum size=0.8em]
    \node (p1) at (2,2.5) {$\bfp_{j-1}$} ;
    \node (p2) at (4,1.5) {$\bfp_{j}$} ;
    \node (p3) at (6,0.5) {$\bfp_{j+1}$} ;
    \node (w1) at (-2,2.5) {$\bfw_{j-1}$} ;
    \node (w2) at (6,-0.5) {$\bfw_{j+2}$} ;
    \node (w3) at (1,1.7) {$\bfw_{j}$} ;
    \node (w4) at (3,0.7) {$\bfw_{j+1}$} ;
    \node (l1) at (-2,1.5) {$\bfu_{j-1}$} ;
    \node (l2) at (0,0.5) {$\bfu_{j}$} ;
    \node (l3) at (2,-0.5) {$\bfu_{j+1}$} ;
    \node (label51) at (-1,3.5) {$\ddots$} ;
    \node (label52) at (5,-1) {$\ddots$} ;
      \draw (p1)-- ++(-1.65,0)  (p2)-- ++(-1.65,0)  (p3)-- ++(-1.65,0) ;
     \draw (w1)-- ++(1.65,0) (2,1.5)-- ++(-1.65,0)  (4,0.5)-- ++(-1.65,0)  (w2)-- ++(-1.65,0) ;
     \draw (l1) -- ++(1.65,0) (l2) -- ++(1.65,0) (l3) -- ++(1.65,0)  ;
    \node[checknode] (cn1) at (0,2) {$T$};
    \node[checknode] (cn2) at (2,1) {$T$};
    \node[checknode] (cn3) at (4,0) {$T$};
   \end{tikzpicture}
 }
\vspace{-0.5cm}
  \caption{
Circuit diagram of a convolutional encoder with a seed transformation matrix $T$.
  }\label{fig:convolutional encoder}
  \vspace{-0.5cm}
\end{figure}
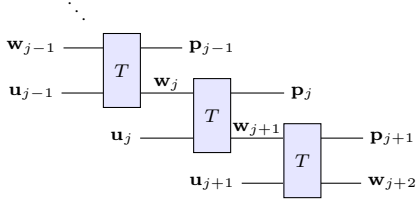
Let  $\bfu_j$ and $\bfp_j$ denote the logical input and physical output symbols at time $j$, respectively,  and let
$\bfw_{j}$ and $\bfw_{j+1}$ denote the $m$ memory symbols before and after  encoding at time $j$, respectively.
Suppose $(A,B,C,E)$ is a realization of $G(D)$~\cite{Forney73} so that
\begin{align}
\bfw_{j+1}=& \bfw_j A+ \bfu_j B, \label{eq:state1}\\
\bfp_j=& \bfw_j C + \bfu_j E, \label{eq:state2}\\
G(D)=&B(D^{-1}I_{m}-A)^{-1}C+E, \label{eq:G_from_ABCE}
 \end{align}
 where $I_{m}$ is the $m\times m$ identity matrix and $A,$$B$,$C$,$D$ are matrices over $\mF_q$ of appropriate dimensions.
Let
\begin{align}
T=\left(\begin{array}{cc} C & A\\  E&B \end{array}\right).  \label{eq:T}
\end{align}
Then (\ref{eq:state1}) and  (\ref{eq:state2}) can be written as
\begin{align} \label{eq:seed transformation encoder}
(\bfp_j:\bfw_{j+1})=(\bfw_{j}:\bfu_j)T.
\end{align}
Thus $T$ is called a \emph{seed transformation matrix}. We will see later that quantum convolutional codes are similarly  defined by a\emph{ seed transformation unitary operator}~\cite{PTO09}.
%
%
%

Forney introduced the idea of \emph{constraint codes}~\cite{Forney11}, which can be used to define
the dual code of a CC.

\bd \label{def:constraint code}
The \emph{constraint  codes} $\mathcal{C}_{(j)}$ of $\cC$ is a $[2m+n,m+k]$ linear block code over $\mathbb{F}_q$ given by
\begin{align}
\cC_{(j)}=&\{ (\bfw_{j}:\bfp_j:\bfw_{j+1})\in \mF_q^{n+2m}: \notag\\
&(\bfp_j:\bfw_{j+1})=(\bfw_{j}:\bfu_j)T,\text{ for }\bfw_j\in \mF_q^m, \bfu_j\in \mF_q^k  \},
\end{align}
where $T$ is a seed transformation matrix of $\cC$ defined in (\ref{eq:T}).
\ed
\noindent
Alternatively, the constraint code $\mathcal{C}_{(j)}$ has a generator matrix
\begin{equation}
\label{eq:TG}
\widetilde{G}=\left(\begin{array}{c|cc}I_{m}&C&A\\ \hline 0& E&B \end{array}\right),
\end{equation}
where $I_m$ is because the input $\bfw_j$ appears in the codeword,  like a systematic encoder; $A,B,C$ and $E$  follow  (\ref{eq:state1}) and  (\ref{eq:state2}).

In general the constraint codes ${\cal C}_{(j)}$ are fixed for all $j$ until termination, where subcodes of ${\cal C}_{(j)}$ are used to have all paths converge to the final zero state.

The dual code ${\cal C}^\perp$ of $\cal C$ is defined\footnote{ The duality is defined by $G(D)H(D)^T=0$ in \cite{GLS09}.
Although not addressed, the duality in \cite{Forney11} is  $G(D)H(D^{-1})^{\bot} = 0$, which is what we adopted in this paper.
This explains the additional transpose on the WAM in the formula for MacWilliams identity in \cite{GLS09}.
 One could also prove the MacWilliams identity from the duality $G(D)H(D)^{\top} = 0$ by modifying Lemma~\ref{def_dual_CCC} appropriately
 to obtain the identity by Gluesing-Luerssen and Schneider in \cite{GLS08,GLS09}.  However, these points are not addressed in \cite{Forney11} (and perhaps unlikely to be addressed by the method in \cite{Forney11}).} as
\begin{align*}
{\cal C}^\perp \ :=& \ \left\{ \bfc'(D) \in \left( \F_q[D] \right)^n \ : \ \bfc(D)\bfc'(D^{-1})^T  = 0,\right.  \\
&\left.\forall \tm{ $\bfc(D) \in {\cal C}$} \right\}.
\end{align*}
The dual code $\cC^{\perp}$ can be generated by a minimal encoder with the same overall constraint length $m$ of $\cC$~\cite{Forney70},
which is important in the following lemma
on the  constraint codes for the dual of a CC.
\bl\label{def_dual_CCC}
Suppose $\cC_{(j)}$ is the constraint code of a minimal encoder $G(D)$.
Let \begin{align}
\widehat{\cal C}_{(j)} =& \left\{ (\bfw_j':\bfp_j':\bfw_{j+1}')  \in \mF_q^{n+2m} \ : \bfw_j'\bfw_j^{T} +\bfp_j'\bfp_j^T \right. \notag\\
&\left.- \bfw_{j+1}'\bfw_{j+1}  =0, \forall \tm{$(\bfw_j:\bfp_j:\bfw_{j+1})  \in {\cal C}_{(j)}$}\right\}.
\end{align}
Then  $\widehat{\cC}_{(j)}$ is a $[2m+n, m+n-k]$ constraint code  of the dual code $\cC^{\perp}$.
In other words, $\widehat{\cC}_{(j)}$
has an $(m+n-k)\times (2m+n)$ generator matrix $\widetilde{H}$ so that $\widetilde{H}'=\widetilde{H} \diag(I_m,I_n,-I_m)$ is a parity-check matrix of $\mathcal{C}_{(j)}$.

\el

\begin{proof} 
 Since $G(D)$ is a minimal encoder, the dual encoder has the same overall constraint length as $G(D)$.
Suppose  $G(D)$ and $H(D)$ are the polynomial generator matrices of $\cC$ and $\cC^{\perp}$, respectively.
Let $\widetilde{H}=\left(\begin{array}{c|cc}I_{m}&C'&A'\\ \hline 0& E'&B' \end{array}\right)$ and $\widetilde{G}$ be as given in (\ref{eq:TG}).
As in (\ref{eq:G_from_ABCE}), the polynomial generator matrix of $\cC^{\perp}$ is $$H(D)=B'(D^{-1}I_{m}-A')^{-1}C'+E'.$$
Since $\widetilde{G}\widetilde{H}'^T=0$, we have
\begin{align}
I_{m}+CC'^T-AA'^T=&0 \label{eq:orthogonality convolutional matrix1}\\
EE'^T-BB'^T=&0 \label{eq:orthogonality convolutional matrix2}\\
CE'^T-AB'^T=&0 \label{eq:orthogonality convolutional matrix3}\\
EC'^T-BA'^T=&0.\label{eq:orthogonality convolutional matrix4}
\end{align}
Since  a minimal encoder is non-catastrophic and delay-free~\cite{Forney70}, it is sufficient to prove the duality by showing that $G(D)H(D^{-1})^T=0$ as follows:
\begin{align*}
&G(D)H(D^{-1})^T\\
=&\left( B(D^{-1}I_m-A)^{-1}C+E  \right)\left( B'(DI_m-A')^{-1}C'+E'  \right)^T\\
\stackrel{(a)}{=}&
B\left( (D^{-1}I_m-A)^{-1}CC'^T((DI_m-A')^{-1})^T\right.\\&\left.+(D^{-1}I_m-A)^{-1}A+A'^T((DI_m-A')^{-1})^T+I_m\right)B'^T\\
\stackrel{(b)}{=}&B\left( (D^{-1}I_m-A)^{-1}AA'^T((DI_m-A')^{-1})^T\right.\\
&+(D^{-1}I_m-A)^{-1}A+A'^T((DI_m-A')^{-1})^T+I_m\\
&\left. -(D^{-1}I_m-A)^{-1}((DI_m-A')^{-1})^T\right)B'^T\\
\stackrel{(c)}{=}&0,
\end{align*}
where $(a)$ follows from (\ref{eq:orthogonality convolutional matrix2}), (\ref{eq:orthogonality convolutional matrix3}), and (\ref{eq:orthogonality convolutional matrix4});
$(b)$ follows from (\ref{eq:orthogonality convolutional matrix1});
$(c)$ is obtained by sequentially applying the following two equalities:
\begin{align*}
&I_{m}+(D^{-1}I_{m}-A)^{-1}A= D^{-1}(D^{-1}I_{m}-A)^{-1}\\
&I_{m}+A'^T((DI_{m}-A')^{-1})^T= D((DI_{m}-A')^{-1})^T.
\end{align*}

\end{proof}

Herein, we clarify the notations of $\widehat{\cC}_{(j)}$ and $\cC_{(j)}^\perp$. We denote by  $\widehat{\cC}_{(j)}$ the constraint code  of the dual CC $\cC^\perp$,
and denote by $\cC_{(j)}^\perp$  the dual of the constraint code $\cC_{(j)}$ of $\cC$.
Specifically, if $(\bfw_{j}:\bfp_j:\bfw_{j+1})=(\bfw_{j}:\bfu_j)\widetilde{H}$ is a codeword of $\widehat{{\cC}}_{(j)}$ for some $ \bfu_j\in\mF_q^{n-k}$,
then  $(\bfw_{j}:\bfp_j: -\bfw_{j+1})$ is a codeword of $\cC_{(j)}^\perp$. 
It is straightforward to prove the following lemma.
\begin{lemma}
\label{lem_EWGFdc}
The EWGF $g_{\widehat{\cC}_{(j)}}^E$ of   $\widehat{\cC}_{(j)}$ is related to the EWGF $g_{\cC_{(j)}^\perp}^E$ of $\cC_{(j)}^\perp$ by
\[
g_{\widehat{\mathcal{C}}_{(j)}}^E= (I^{\otimes m+n}\otimes \Pi) g_{{\mathcal{C}^{\perp}_{(j)}}}^E,
\]
where  $\Pi=\sum_{\bfw\in\mF_q^m} \ket{\bfw}\bra{-\bfw}$ is a permutation on the $m$ memory symbol states
and $I^{\otimes m+n}$ is the identity operator on the first $m+n$ symbol states.
\end{lemma}
\noindent Note that in the case that $q$ is a power of 2, $\Pi$ is trivial and  we have $\widehat{\cC}_{(j)}=\cC^\perp_{(j)}$.

The weight adjacency matrix of $\cC$ is the weight enumeration of $\mathcal{C}_{(j)}$ in matrix form.
\bd \label{def_WAM}
The \emph{weight adjacency matrix} (WAM) $\Lambda_{\cC_{(j)}}(y)$ of a CC with a constraint code $\mathcal{C}_{(j)}$ is the matrix whose
($\bfw_j$, $\bfw_{j+1}$) entry is a HWGF of the output symbols of $\cC_{(j)}$ with the memory symbols $\bfw_j$ and $\bfw_{j+1}$ before and after time $j$, respectively.
That is,
\begin{align*}
\bra{\bfw_j}\Lambda_{\cC_{(j)}}(y)\ket{\bfw_{j+1}}&\equiv( \bra{\bfw_j}\otimes \gamma_H^{\otimes n}(y) \otimes \bra{\bfw_{j+1}}) g^E_{\cC_{(j)}}\\
&=\gamma_H^{\otimes n}(y) \left( \sum_{\bfp_j\in \mF_q^n: \atop{(\bfw_j:\bfp_j:\bfw_{j+1})\in \mathcal{C}_{(j)}}} \ket{\bfp_j}\right).
\end{align*}
\ed

Now we are ready to derive the MacWilliams identity for convolutional codes  \cite{GLS08,GLS09}.
\bt \label{thm:MacWilliamsIdentity_WAM}
Suppose the WAM from a minimal encoder of an $(n,k,m)$ CC $\cC$  over $\mathbb{F}_q$ is $\Lambda_{\cC_{(j)}}(x,y)$.
Then the WAM for  $\widehat{\cC}_{(j)}$ is given by
\begin{align}
\Lambda_{\widehat{\cC}_{(j)}}(y)&= \frac{(1+(q-1)y)^n}{q^{m+k}} \mathcal{F}_{\mF_q^m} \Lambda_{\cC_{(j)}}\left(\frac{1-y}{ 1+(q-1)y}\right) \mathcal{F}_{\mF_q^{m}}^{\dag}. \label{eq:MI for WAM}
\end{align}
\et
\begin{proof}
From Lemma~\ref{lem_EWGFdc}, Theorem~\ref{thm:MacWilliamsIdentity_EWGF}, and Lemma~\ref{lemma:fourier transform},
\begin{align*}
g_{\widehat{\mathcal{C}}_{(j)}}^E
=& \frac{1}{q^{m+k}} \sum_{\bfw_j\in \mF_q^m}\sum_{\bfw_{j+1}\in \mF_q^m} \sum_{\bfp_j\in\mF_q^n: \atop{(\bfw_j:\bfp_j:\bfw_{j+1})\in \mathcal{C}_{(j)}}} \\
&(\cF_{\mF_q^m} \ket{\bfw_j}) (\cF_{\mF_q^n} \ket{\bfp_j}) (\Pi\cF_{\mF_q^m} \ket{\bfw_{j+1}}).
\end{align*}
Thus
\begin{align*}
&\bra{\bfw_j'}\Lambda_{\widehat{\cC}_{(j)}}(y)\ket{\bfw_{j+1}'}=(\bra{\bfw_j'}\otimes \gamma_H^{\otimes n}(y) \otimes \bra{\bfw_{j+1}'}) g^E_{\widehat{\cC}_{(j)}}\\
 =& \frac{1}{q^{m+k}} \sum_{\bfw_j} \sum_{\bfw_{j+1}}  \bra{\bfw_{j}'}\cF_{\mF_q^m} \ket{\bfw_j} \\
  &\cdot\left( \sum_{\bfp_j\in \mF_q^n: \atop{(\bfw_j:\bfp_j:\bfw_{j+1})\in \mathcal{C}_{(j)}}}  \gamma_H^{\otimes n}(y) \cF_{\mF_q^n} \ket{\bfp_j}\right)   \bra{\bfw_{j+1}'}\Pi \cF_{\mF_q^m}\ket{\bfw_{j+1}} \\
 \stackrel{(a)}{=}& \frac{(1+(q-1)y)^n}{q^{m+k}} \sum_{\bfw_j} \sum_{\bfw_{j+1}} \bra{\bfw_{j}'}\cF_{\mF_q^m} \ket{\bfw_j} \\
 &\bra{\bfw_{j}} \Lambda_{\cC_{(j)}}\left(\frac{1-y}{1+(q-1)y}\right)\ket{\bfw_{j+1}}  \bra{\bfw_{j+1}} \cF_{\mF_q^m}^{\dag} \ket{\bfw_{j+1}'} \\
 \stackrel{(b)}{=}& \frac{(1+(q-1)y)^n}{q^{m+k}}  \bra{\bfw_{j}'}\cF_{\mF_q^m}  \Lambda_{\cC_{(j)}}\left(\frac{1-y}{1+(q-1)y}\right) \cF_{\mF_q^m}^{\dag}  \ket{\bfw_{j+1}'},
\end{align*}
where $(a)$ follows from the definition of $\Lambda_{\cC_{(j)}}(y)$ and $\bra{\bfw_{j+1}'}\Pi \cF_{\mF_q^m}  \ket{\bfw_{j+1}}=\bra{\bfw_{j+1}} \cF_{\mF_q^m}^{\dag}  \ket{\bfw_{j+1}'}$;
$(b)$ follows from $\sum_{\bfw_j}\ket{\bfw_j}\bra{\bfw_j}=\sum_{\bfw_{j+1}}\ket{\bfw_{j+1}}\bra{\bfw_{j+1}}=I^{\otimes m}$.
The equation $\bra{\bfw_{j+1}'}\Pi \cF_{\mF_q^m}  \ket{\bfw_{j+1}}$ $= \bra{\bfw_{j+1}} \cF_{\mF_q^m}^{\dag}  \ket{\bfw_{j+1}'}$
can be  verified straightforwardly.
Then (\ref{eq:MI for WAM}) follows directly.
 \end{proof}
Thus we have provided a  more direct and transparent proof than that in \cite{GLS08,GLS09}. 
 Also, this direct proof gives the scalar in the MacWilliams identity for CCs that is not shown  in~\cite{Forney11}.

\textbf{Remark:} the assumption of minimal encoder in this section could be removed by using the componentwise duality definition of constraint codes and the normal factor graph duality theorem~\cite{Forney01,Forney11}.
That is, the normal  graph duality theorem implies the definition of dual constraint codes in Lemma~\ref{def_dual_CCC} without the requirement of minimality.
For the rest, we will assume that the dual constraint codes are determined without the encoder being minimal.

In~\cite{LHL14b}, we have this theorem proved using the usual vector notation, instead of the Dirac notation here.

\subsection{The MacWilliams Identity for the Input-Output Weight Adjacency Matrices }\label{sec_beyondGL}

As Gluesing-Luerssen and Schneider noted in \cite{GLS09}, the input-output weight generating functions are not invariants of a CC, but rather of the encoder.
We will derive the MacWilliams identity for IOWAMs 
for systematic encoders.

Recall that  a seed transformation matrix of an $(n,k,m)$ convolutional code $\cC$ is of the form $T=\left(\begin{array}{cc} C & A\\  E&B \end{array}\right)$
and it defines the constraint code with a generator matrix $\widetilde{G}$  given in (\ref{eq:TG}).
Here we consider the seed transformation matrix of a \emph{systematic encoder}, that is, the matrices $C$ and $E$  are in the following specific form:
$
\left(\begin{array}{c} C \\ E \end{array}\right) = \left( \begin{array}{cc} 0 & C_0 \\ I_k & E_0\end{array}\right),
$
where $C_0$ and $E_0$ are $m\times (n-k)$ and $k\times (n-k)$ matrices, respectively.
We may assume the generator matrix  $\widetilde{G}_S$ of the constraint code corresponding to a systematic encoder  is
\begin{equation}\label{eq:systemG}
\widetilde{G}_S=\left(\begin{array}{c|cc|c}I_{m}&0&C_0&A_0 \\ \hline 0&I_k& E_0&B_0 \end{array}\right).
\end{equation}
Thus $(\bfw_j,\bfp_j,\bfw_{j+1})$ is a codeword of $\cC_{(j)}$ for $\bfu_j\in \mF_q^k$ if
\begin{align*}
\bfw_{j+1}&= \bfw_j A_0+\bfu_j B_0,\\
\bfp_j&=(\bfu_j : \bfw_j C_0+ \bfu_j E_0)
\triangleq \left(\bfp^I:\bfp^P\right),
\end{align*}
where  $\bfp^I=\bfu_j$ and $\bfp^P=  \bfw_j C_0+ \bfu_j E_0$.

\bd \label{def:IP WAM}
The input-parity weight adjacency matrix (IPWAM) of a \emph{systematic} convolutional encoder $\widetilde{G}_S$ is the matrix $\Lambda_{\cC_{(j)}}^{IP}(x,y)$ whose ($\bfw_j$, $\bfw_{j+1}$) entry is an IPWGF of $\cC_{(j)}$ in $x$ and $y$ by $\sum_{i,o} \nu_{i,o} x^i y^o$,
where $\nu_{i,o}$ is the number of $\bfp_j=\left(\bfp^I:\bfp^P\right)\in\mF_q^n$  with $\wt{\bfp^I}=i$ and $\wt{\bfp^P}=o$ so that $(\bfw_j:\bfp^I:\bfp^P:\bfw_{j+1})\in \mathcal{C}_{(j)}$.
That is,
\begin{align*}
&\bra{\bfw_j}\Lambda_{\cC_{(j)}}^{IP}(x,y)\ket{\bfw_{j+1}}\\
=&\sum_{\bfp_j\in \mF_q^n: \atop{(\bfw_j:\bfp_j:\bfw_{j+1})\in \mathcal{C}_{(j)}}}  \gamma_H^{\otimes k}(x)\otimes \gamma_H^{\otimes n-k}(y) \ket{\bfp_j},
\end{align*}
where $\gamma_H$ is defined in (\ref{eq:gamma Hamming}).
\ed

For simplicity, we assume the corresponding systematic encoder $\widetilde{H}_S$ for the constraint code $\widehat{C}_{(j)}$ of $\cC^{\perp}$ is of the form
$
\widetilde{H}_S=\left(\begin{array}{c|cc|c}I_m &C_0'&0& A_0'\\ \hline 0&E_0'& I_{n-k} &B_0' \end{array}\right).
$
Thus $(\bfw_j,\bfp_j,\bfw_{j+1})$ is a codeword of $\widehat{\cC}_{(j)}$ for $\bfu_j\in \mF_q^{n-k}$ if
\begin{align*}
\bfw_{j+1}&= \bfw_j A_0'+\bfu_j B_0',\\
\bfp_j&=( \bfw_j C_0'+ \bfu_j E_0': \bfu_j)\triangleq \left(\bfp^P:\bfp^I\right).
\end{align*}
where  $\bfp^I=\bfu_j$ and $\bfp^P=  \bfw_j C_0'+ \bfu_j E_0'$.

Similarly to the previous development, it is straightforward to have
the following MacWilliams identity for the IPWAMs of a systematic encoder of $\cC$ and its dual.

\bt \label{thm:MacWilliamsIdentity_IO WAM}
Suppose the IPWAM of a systematic minimal encoder of an $(n,k,m)$ CC $\cC$ over $\mathbb{F}_q$  is $\Lambda_{\cC_{(j)}}^{IP}(x,y)$.
Then the IPWAM of its dual encoder of $\cC^{\perp}$  is
\begin{align}
&\Lambda_{\widehat{\cC}_{(j)}}^{IP}(x,y)\notag\\
=&\frac{(1+y)^k(1+x)^{n-k}}{2^m}\cF_{\mathbb{F}_2}^{\otimes m}\Lambda_{\cC_{(j)}}^{IP}(\frac{1-y}{1+y},\frac{1-x}{1+x}) \cF_{\mathbb{F}_2}^{\otimes m}, \label{thm:8_revised}
\end{align} \label{eq:IOWAM}

\et

Remark: Given the systematic encoder $\widetilde{G}_S$ of $\cC$ in (\ref{eq:systemG}), it may be natural to define the corresponding systematic encoder $\widetilde{H}_S$ for the constraint code $\widehat{C}_{(j)}$ of $\cC^{\perp}$
as
\begin{equation*}\label{eq:dualHS}
\widetilde{H}_S=\left(\begin{array}{c|cc|c}C_0^T&E_0^T&I_{n-k}&0\\ \hline A_0^T&B_0^T&0&-I_m \end{array}\right).
\end{equation*}
Then a codeword of $\widehat{C}_{(j)}$ is in the reverse order $(\bfw_{j+1}:\bfp_j:\bfw_{j})=(\bfw_{j+1}:\bfp^P:\bfp^I:\bfw_{j})$.
The MacWilliams identity for the IPWAMs in the above theorem still holds, except that $\Lambda_{\widehat{\cC}_{(j)}}^{IP}(x,y)$ is replaced by its transpose 
in (\ref{eq:IOWAM}).
However, this identity for CCs  would require  the duality $G(D)H(D)^T=0$,
which is inconsistent with our development.
\be \label{ex:ip_WAM}
Consider the constraint code  $\cC_{(j)}$ of an $(n=2,k=1,m=2)$ CC over $\mathbb{F}_2$ with the following generator and parity-check matrices in the systematic form:
\[
\tilde{G}=\left(\begin{array}{c|c|c|c}
10&0&1&01\\
01&0&0&10\\
00&1&1&10\\
\end{array}\right)
\text{ and }
\tilde{H}=\left(\begin{array}{c|c|c|c}
10&1&1&00\\
01&1&0&10\\
10&0&0&01\\
\end{array}\right).
\]
In the case of binary codes, Lemma~\ref{def_dual_CCC} says that the dual constraint code $\hat{\cC}_{(j)}$  is simply the dual code of $\cC_{(j)}$, and is generated by $\tilde{H}$.
The Fourier transform operator over $\mF_2$ is
\[
\cF_{\mathbb{F}_2}=\begin{bmatrix} 1&1\\1&-1\end{bmatrix}=\cF_{\mF_2}^{\dag}.
\]
By Definition~\ref{def:IP WAM},
\begin{align} \label{eq:31}
\Lambda_{\cC_{(j)}}^{IP}(x,y)=& \left(\begin{array}{cccc}
1&xy&0&0\\
0&0&y&x\\
xy&1&0&0\\
0&0&x&y\end{array}\right) \end{align}
 in the ordered basis $\{\ket{00},\ket{10},\ket{01},\ket{11}\}$.
By the MacWilliams identity (\ref{thm:8_revised}), we have, 
\begin{align*}
\Lambda_{\widehat{\cC}_{(j)}}^{IP}(x,y)
=& \left(\begin{array}{cccc}
1&0&xy&0\\
xy&0&1&0\\
0&y&0&x\\
0&x&0&y\end{array}\right),
\end{align*}
which is the same as determined from $\hat{\cC}_{(j)}$ by Definition~\ref{def:IP WAM}.
\eep

\subsection{Weight Enumeration on the Full Trellis Diagram for Convolutional Codes} \label{subsec:remarks on free distance}
In this subsection we will  focus mainly on the Hamming weight enumeration on the codewords of a CC and consider the weight enumeration in one indeterminate since codeword length is not fixed.
We reconsider enumerating walks on the full trellis diagram\footnote{By the full trellis diagram of $\cal C$ we mean the trellis diagram of $\cal C$ with arbitrary beginning and ending states.} of $\cC$ in a matrix.
 That is, we define
\beq \label{eq_wholetrellis00}
\Lambda_{{\cal C}}(y,D) := \left( {I_{q^m}} - \Lambda_{\cC_{(j)}}(y) D \right)^{-1} = \sum_{d \geq 0} \left(\Lambda_{\cC_{(j)}}(y)\right)^d D^d,
\eeq
where the $(\bfw,\bfw')$ entry of the matrix $\left(\Lambda_{\cC_{(j)}}(y)\right)^d$ is the enumeration of the Hamming weights of length-$d$ walks that begin at state $\bfw$ at time $0$ and end at state $\bfw'$ at time $d$ on the full trellis diagram of $\cal C$.
Similarly, 
we define the full-trellis weight generating function (FWGF)
$$
\Lambda_{{\cal C}^\perp}(y,D) = \left( {I_{q^m}} - \Lambda_{\widehat{\cC}_{(j)}}(y) D \right)^{-1}. 
$$
Note that we assume the constraint codes ${\cal C}_{(j)}$  (and $\widehat{\cC}_{(j)}$)
are the same for all $j$. Then we have the following identity.

\begin{theorem} \label{thm:9} 
The FWGFs of a CC $\cC$ and its dual are related by the following MacWilliams identity:
\begin{align}
&\Lambda_{{\cal C}^\perp}(y,D) \notag\\
=&  \frac{1}{q^m} {\cF}_{\mF_q^m}\Lambda_{{\cal C}}\left( \frac{1-y}{1+(q-1)y},\ \frac{(1+(q-1)y)^n}{q^{k}}D\right) {\cF}_{\mF_q^m}^\dag. \label{eq:15}
\end{align}
\end{theorem}
\begin{proof}
It follows directly from the definitions of $\Lambda_{{\cal C}}(y,D)$, $\Lambda_{{\cal C}^\perp}(y,D)$ and Theorem \ref{thm:MacWilliamsIdentity_IO WAM} that 
{\small
\begin{align*}
&\Lambda_{{\cal C}^\perp}(y,D) =  \left( {I_{q^m}} - {\Lambda}_{\widehat{\cC}_{(j)}}(y) D \right)^{-1} \\
=&    \left[ { I_{q^m}} - \frac{(1+(q-1)y)^n}{q^{m+k}} {\cF}_{\mF_q^m}  \Lambda_{\cC_{(j)}}\left(  \frac{1-y}{1+(q-1)y}\right) {\cF}_{\mF_q^m}^\dag D \right]^{-1} \\
=& \frac{1}{q^m}   {\cF}_{\mF_q^m} \left[ {I_{q^m}} - \frac{(1+(q-1)y)^n}{q^{k}} \Lambda_{\cC_{(j)}}\left(  \frac{1-y}{1+(q-1)y}\right) D\right]^{-1} {\cF}_{\mF_q^m} ^\dag \\
=& \frac{1}{q^m} {\cF}_{\mF_q^m}\Lambda_{{\cal C}}\left( \frac{1-y}{1+(q-1)y},\ \frac{(1+(q-1)y)^n}{q^{k}}D\right) {\cF}_{\mF_q^m}^\dag.
\end{align*}
}
\end{proof}

Finally, 
we  state the duality result for full-trellis IPWGF for a systematic CC $\cal C$. 
For any codeword $\bfc \in {\cal C}$, let the weight function of $\bfc$ be
\[
f_{D}^{\tm{IP}} (\bfc)= x^{\sum_{i=0}^{\deg(\bfc)} \wt{\bfc_{i,I}}} y^{\sum_{i=0}^{\deg(\bfc)}\wt{\bfc_{i,P}}} D^{\deg(\bfc)},
\]
where $\bfc_i = \left( \bfc_{i,I}:  \bfc_{i,P} \right) $ is defined as in Example~\ref{ex:ip}.
Suppose $\cC$ is a systematic CC and ${\cal C}^\perp$ is the dual code of $\cal C$. Let  
\beq \label{eq_wholetrellis00IP}
\Lambda_{{\cal C}}^{IP}(x,y,D) := \left( {I_{q^m}} - \Lambda_{\cC_{(j)}}^{IP}(x,y) D \right)^{-1}
\eeq
be the full-trellis input-parity weight adjacency matrix (FIPWAM) for walks on the full trellis diagram of $\cal C$ defined with respect to the above weight function.
Along the same way, we have the following theorem.

\bt \label{thm:20}

The FIPWAMs of a systematic CC $\cal C$ and its dual are related by the following MacWilliams identity:
{
\begin{align*}
&\Lambda_{{\cal C}^\perp}^{IP}(x, y,D) = \frac{1}{q^m} {\cF}_{\mF_q^m}\Lambda_{{\cal C}}^{IP}\left( \frac{1-y}{1+(q-1)y},\right.\\
&\left.  \frac{1-x}{1+(q-1)x}, \frac{(1+(q-1)y)^k(1+(q-1)x)^{n-k}}{q^{k}}D\right) {\cF}_{\mF_q^m}^\dag.
\end{align*}
}
\etp

\subsection{Relation Diagram}

The relation of this weight enumeration and the free distance enumeration of a CC has been discussed in \cite{McE98,GL05}.
Here we discuss further with the introduction of MacWilliams identity.

Let $W_N^{IP}(x,y)$ denote the input-parity weight enumerator in $x,y$ for the paths of length $N$ and whose initial and final states are both $\bf 0$.
\label{def_twgf}The \emph{total input-parity weight generating function} (TIPWGF) $W_{{\cal C}}(x,y,D)$ of a CC  is 
\begin{align*}
W_{{\cal C}}^{IP}(x,y,D) = \ \sum_{N\geq 0} W_N^{IP}(x,y) D^N.
\end{align*}
From \cite[Theorem 4.1]{McE98}, we have
\begin{align}
W_{{\cal C}}^{IP}(x,y,D)  &=  \bra{{\bf 0}} \Lambda_{\cC}^{IP}(x,y,D) \ket{{\bf 0}} \label{eq:5} \\
=& \bra{{\bf 0}} \left( { I_{q^m}} - \Lambda_{\cC_{(j)}}^{IP}(x,y) D \right)^{-1} \ket{{\bf 0}}, \label{eq:36}
\end{align}
where $\ket{\bf{0}}$ is the zero vector in $\mF_q^m$ and $\bra{{\bf 0}}M \ket{{\bf 0}}$ represents the $(0,0)$-entry of a  matrix $M$.

The input-parity free distance  enumerator of a CC $W_{\cC_\text{free}}^{IP}(x,y,D)$ is similarly defined but it counts the paths which do not enter the zero state except at the beginning and end.
Since the contribution of the path from state $(\bf 0,0)$ to state $(\bf 0,0)$ is unwanted, $\left( \Lambda_{\cC_{(j)}}^{IP}(x,y) - \ket{{\bf 0}}\bra{{\bf 0}} \right)$ is considered instead.
From \cite[Theorem 3.1]{McE98} (see also \cite{GL05}), we have
\beq
W_{{\cC_\tm{free}}}^{IP} (x,y,D) \ = 1-\frac{1}{\bra{{\bf 0}} \left( {I_{q^m}} - \left( \Lambda_{\cC_{(j)}}^{IP}(x,y) - \ket{{\bf 0}}\bra{{\bf 0}} \right) D \right)^{-1} \ket{{\bf 0}}}. \label{eq:wcdiv}
\eeq
\begin{theorem}
\label{cor:9}
The TIPWGF $W_{{\cal C}}^{IP}(x,y,D)$ and the input-parity free distance  enumerator $W_{\cC_\text{free}}^{IP}(x,y,D)$ can be derived from each other:
\begin{align}
W_{\cC_\tm{free}}^{IP} (x,y,D) &=1- \frac{1+W_{{\cal C}}^{IP}(x,y,D) D}{W_{{\cal C}}^{IP}(x,y,D)}.\label{eq:21}
\end{align}
\end{theorem}
\begin{proof}
Rewrite \eqref{eq:wcdiv} as
\begin{align*}
&\frac{1}{1-W_{\cC_\tm{free}}^{IP} (x,y,D)} \\
 =& \bra{{\bf 0}} \left[ \left( { I_{q^m}} - \Lambda_{\cC_{(j)}}^{IP}(x,y) D \right) + \ket{{\bf 0}}\bra{{\bf 0}} D \right]^{-1} \ket{{\bf 0}},
\end{align*}
and note that $W_{\cal C}(x,y,D)=\bra{{\bf 0}}\left( {I_{q^m}} - \Lambda_{{\cal C}_{(j)}}^{IP}(x,y)  D\right)^{-1} \ket{{\bf 0}}$.
By Woodbury identity for matrix inverse, we have (\ref{eq:19}),
\begin{figure*}[!t]
\begin{align}
\lefteqn{\left[ \left( { I_{q^m}} - \Lambda_{\cC_{(j)}}^{IP}(x,y) D \right) + \ket{{\bf 0}}\bra{{\bf 0}} D \right]^{-1}}\notag\\
&=& \left( { I_{q^m}} - \Lambda_{\cC_{(j)}}^{IP}(x,y) D \right) ^{-1} - \frac{D}{1 + \bra{{\bf 0}} \left( { I_{q^m}} - \Lambda_{\cC_{(j)}}^{IP}(x,y) D \right)^{-1}  \ket{{\bf 0}}D} \left( {I_{q^m}} - \Lambda_{\cC_{(j)}}^{IP}(x,y) D \right)^{-1} \ket{{\bf 0}}\bra{{\bf 0}}
\left( { I_{q^m}} - \Lambda_{\cC_{(j)}}^{IP}(x,y) D \right)^{-1} \label{eq:19}
\end{align}
$\overline{\hspace{\textwidth}}$
\end{figure*}
which in turn gives the desired equation by considering the $(\bf 0,0)$-entry of both sides.

\end{proof}
We summarize the relations\footnote{We thank anonymous referees for pointing out some mistakes in our previous results \cite[Sec.~V]{LHL14b}.
More precisely,   we made a mistake on counting the free distance paths when deriving Theorem~\ref{cor:9}:
 we were counting the so-called ``molecular paths" while we claimed to count the ``atomic paths" \cite{McE98}. Thus (16)-(18) in \cite[Sec.~V]{LHL14b}  and the discussion there are incorrect. }
of these weight enumerators in Fig. \ref{fig:relations_diagram}.
Shearer and McEliece  have observed more than 40 years ago  that the MacWilliams identity does not exist between $W_{\cC_\tm{free}}^{IP} (y,y,D)$ and $W_{{\cC}_\tm{free}^\perp}^{IP} (y,y,D)$ \cite{SM77}. We have an explanation from Figure~\ref{fig:relations_diagram}: This is because there is a one-one correspondence, i.e. a duality, between $\Lambda_{{\cal C}}^{IP}(x,y,D)$ and $\Lambda_{{\cal C}^\perp}^{IP}(x,y,D)$ but not for $W_{{\cal C}}^{IP}(x,y,D)$ and $W_{{\cal C}^\perp}^{IP}(x,y,D)$, since $W_{{\cal C}}^{IP}(x,y,D)$ corresponds  to only the entry of $\Lambda_{{\cal C}}^{IP}(x,y,D)$ associated with $({\bf 0},{\bf 0})$-entry as shown in  (\ref{eq:5}).
We point out that, in~\cite{Forney11}, Forney considered terminated CCs and defined a normalized Hamming weight distribution $\bar{g}_{\cC}(y)=\lim_{N\rightarrow \infty} \text{Tr }{\Lambda^N(y)} $, which has a MacWilliams identity, since for finite $N$, terminated convolutional codes are linear block codes.
He showed, by examples, that the free distance enumerator converges to $\bar{g}_\cC(y)$ rapidly.

\begin{figure}
\[
\xymatrix{
&\Lambda_{{\cC}_{(j)}}^{IP}(x,y)\ar[d]^{\text{Eq.}~(\ref{eq_wholetrellis00IP})} \ar[r]^{\text{Theorem}~\ref{thm:MacWilliamsIdentity_IO WAM}} &\Lambda_{\hat{\cal C}_{(j)}}^{IP}(x,y) \ar[l]\ar[d]^{\text{Eq.}~(\ref{eq_wholetrellis00IP})}&\\
 &  \Lambda_{{\cal C}}^{IP}(x,y,D)  \ar[d]^{\text{Eq.}~(\ref{eq:5})} \ar[r]^{\text{Theorem}~\ref{thm:20}}  & \Lambda_{{\cal C}^\perp}^{IP}(x,y,D) \ar[l] \ar[d]^{\text{Eq.}~(\ref{eq:5})}  \\
& W_{{\cal C}}^{IP}(x,y,D) \ar[d]^{\text{Theorem}~\ref{cor:9}} &W_{{\cal C}^\perp}^{IP}(x,y,D)\ar[d]^{\text{Theorem}~\ref{cor:9}} \\
 & W_{\cC_\tm{free}}^{IP} (x,y,D) \ar[u] & W_{\cC_\tm{free}^\perp}^{IP}(x,y,D) \ar[u]
}
\]
\caption{Relation diagram of various weight enumerations. A relation $A \rightarrow B$ means that $B$ can be derived given $A$.
When $x=y$, we have the relations for Hamming weight generating functions.} \label{fig:relations_diagram}
\end{figure}
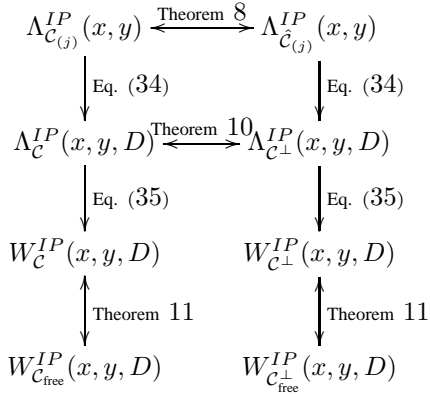

The following corollary can be obtained  by combining Theorem~\ref{thm:20} and Eq.~(\ref{eq:5}). 
\begin{corollary} \label{eq:cor6}
The  TIPWGF of the dual code of a CC  can be determined
from its FIPWAM (\ref{eq_wholetrellis00}):
\begin{align*}
&W_{{\cC}^\perp}^{IP}(x,y,D)= \bra{{\bf 0}} \frac{1}{q^m} {\cF}_{\mF_q^m}\Lambda_{{\cal C}}^{IP}\left( \frac{1-y}{1+(q-1)y},\right.\\
&\left.  \frac{1-x}{1+(q-1)x}, \frac{(1+(q-1)y)^k(1+(q-1)x)^{n-k}}{q^{k}}D\right) {\cF}_{\mF_q^m}^\dag\ket{{\bf 0}}.
\end{align*}
\end{corollary}
These results can be directly generalized to input-output free distance enumerators. Similar to (\ref{eq:14}), let
\begin{align}
W_{{\cC_{\tm{free}}}}^{IO}(x,y,D)=&  W_{{\cC_{\tm{free}}}}^{IP}(xy,y,D) \label{eq:39}  
\end{align}
The bit error rate (BER) $P_b$ of a CC over a $q$-ary symmetric channel with error rate $\epsilon$ can be upper bounded by $W_{{\cal C}_{\tm{free}}}^{IO}(x,y,D)$ as follows~\cite{McE02}:
\begin{align}
P_b\leq \left.\frac{1}{k}\frac{\partial W_{{\cal C}_{\tm{free}}}^{IO}(x,y,D)}{\partial x}\right|_{x=D=1,y=\beta}, \label{eq:41}
\end{align}
where $\beta= 2\sqrt{\epsilon (1-(q-1)\epsilon)}+\epsilon(q-2)$.\\

\noindent\textbf{Example~\ref{ex:ip_WAM}.} (Continued.)

By (\ref{eq:39}), (\ref{eq:36}), Theorem~\ref{cor:9} and $\Lambda_{{\cC}_{(j)}}^{IP}(x,y)$ in (\ref{eq:31}),   we have
\begin{align}W_{{\cC_{\tm{free}}}}^{IO}(x,y,D)=& \frac{D^3 x^2 y^5 (1 + D (-1 + x^2) y)}{1 - D y - D^2 y -
  D^3 (-1 + x^2) y^2} \notag \\
=&  W_{{\cC_{\tm{free}^{\perp}}}}^{IO}(x,y,D),
\end{align}
where the MacWilliams identity in Theorem \ref{thm:MacWilliamsIdentity_IO WAM} or Theorem \ref{thm:20}  is used for the dual code.
Thus by (\ref{eq:41}), the BER of $\cC$ or $\cC^{\perp}$ over a binary symmetric channel with rate $\epsilon$ is bounded by
\[
P_b\leq \frac{2 \beta^5 (1 - \beta - \beta^2)}{(1 - 2 \beta)^2} =2\beta^5+O(\beta^6),
\]
where $\beta=2\sqrt{\epsilon(1-\epsilon)}$.

The BERs of $\cC$ and $\cC^{\perp}$ are shown in Fig.~\ref{fig:cc_sim_plot}, together with the analytical upper bound (\ref{eq:41}).
The Viterbi decoder is used with traceback length (TL) $5$ and $10$.
The BERs of the CC and its dual are very close and they are close to the upper bound at low $p$.
When TL is greater or equal to $5$, the simulated BER curve is completely below the upper bound.
The theoretic upper bound obtained by the input-parity free distance enumerator provides a reasonable performance benchmarking for the CCs.

\begin{figure}[ht]
\[\includegraphics[width=9.0cm]{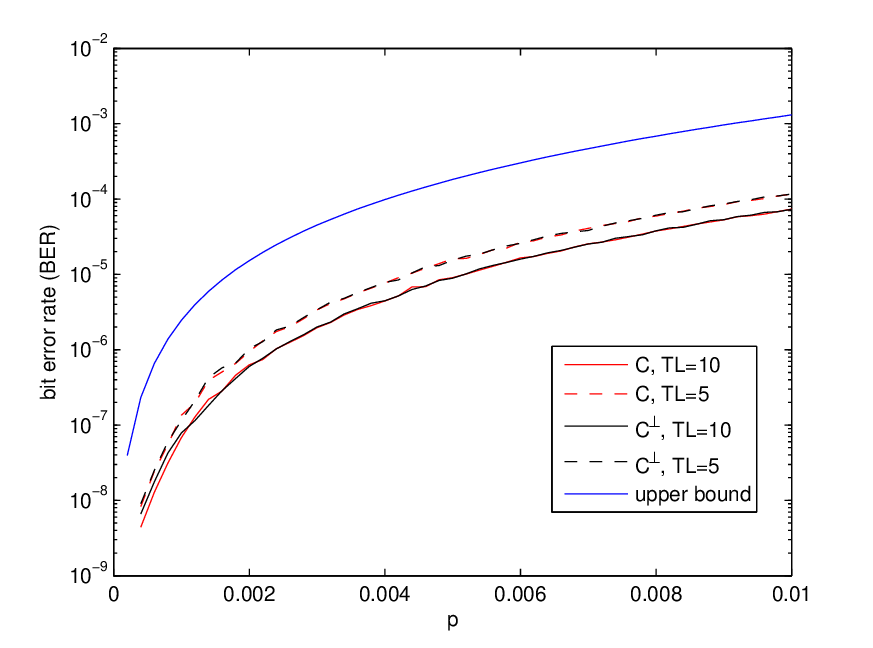}\]
  \caption{Comparision of the BER of $\cC$ or $\cC^{\perp}$ over a binary symmetric channel with rate $p$ and the upper bound (\ref{eq:41}).
 For each $p$, $10^8$ samples are simulated.}
\label{fig:cc_sim_plot}
\end{figure}
\section{The MacWilliams Identity for Quantum Convolutional Codes}
\label{sec:quantum MI}
\subsection{Pauli Operators and Quantum Codes} \label{sec:IntroQECC}

Poulin \emph{et al.} devised a representation of quantum convolutional codes by a seed transformation and  defined the associated state diagram  \cite{PTO09,OT03}.
Entanglement-assisted quantum convolutional codes (EA-QCCs) are developed in \cite{WHZ14}.
We will use this representation to develop the duality in quantum convolutional codes.

Let us start with some notation.
A single-qubit state  is a vector in the complex Hilbert space $\mathbb{C}^2$ and an $n$-qubit state is a vector in $\mC^{2^n}$.
The Pauli matrices
$I=\begin{bmatrix}1 &0\\0&1\end{bmatrix},$ $X=\begin{bmatrix}0 &1\\1&0\end{bmatrix}$, $Z=\begin{bmatrix}1 &0\\0&-1\end{bmatrix},$   and  $Y=iZX$
form a basis of the space of linear operators on $\mC^2$.
Let $\mathcal{G}_1=\{\pm I, \pm iI, \pm X, \pm iX, \pm Y, \pm iY, \pm Z, \pm iZ\}$ be the Pauli group.
and let $\mathcal{G}_n=\mathcal{G}_1^{\otimes n}$ denote the $n$-fold Pauli group.
The weight $\wt{g}$ of $g\in \mathcal{G}_n$ is the number of its components  that are not the identity operator.
Let $Z_i=I^{\otimes i-1}\otimes Z\otimes I^{\otimes n-i}$ and $X_i=I^{\otimes i-1}\otimes X\otimes I^{\otimes n-i} $ be the Pauli operators on the $i$-th qubit for convenience
and the total number of qubits is clear from the context.
For $g, h \in {\mathcal{G}}_n$, the \emph{symplectic inner product} $\langle\ ,\ \rangle_{\cG_n}$  is defined by
\begin{align} \label{eq:Pauli inner product}
\langle g, h\rangle_{\cG_n}= \left\{
               \begin{array}{ll}
                 0, & \text{if }gh-hg=0; \\
                 1, & \text{if }gh+hg=0.
               \end{array}
             \right.
\end{align}

A \emph{seed transformation} $U$ on $\mC^{2^n}$   is a unitary Clifford operator that preserves $\cG_n$ under conjugation \cite{CRSS98,Got97,BDM062}.
Suppose  $\mathcal{S}$ is an Abelian subgroup of  $\mathcal{G}_n$
with a set of $ n-k$ independent generators  defined by a seed transformation $U$
$$\cS=\{ U(I^{\otimes k}\otimes S^Z) U^{\dag}: S^Z\in \{I,Z\}^{\otimes n-k}\}$$
and $\mathcal{S}$ does not include $ -I^{\otimes n}$.
An $[[n,k]]$ quantum stabilizer code $\cC_{\cS}$ is the joint-$(+1)$ eigenspace of  $\cS\subset\cG_n$.
That is, $$\cC_\cS=\{\ket{\psi}\in \mC^{2^n}: g\ket{\psi}=\ket{\psi}, \forall g\in \mathcal{S}\}.$$
$U$ is  called a Clifford encoder of the quantum code.
In this definition we implicitly assume that the first $k$ qubits before encoding are logical qubits.
Let
\begin{align*}
{C}(\mathcal{S})
=&\{U(L\otimes S^{Z})U^{\dag}: S^Z\in\{I,Z\}^{\otimes n-k} , L\in \cG_k \}
\end{align*}
be the centralizer group of $\cS$, consisting of operators  in $\cG_n$ that commute with the stabilizers. 
Thus elements in $C(\cS)$ are logical operators since they operate on the code space.
Note that the orthogonal group of $\cS$  is ${C}(\mathcal{S})$. 

Even though an error can be a linear combination of Pauli operators, it will be \emph{discretized} (in the basis of Pauli operators) after quantum measurements. Also, the overall phase cannot be observed.
Thus it suffices to consider  errors in
\begin{align}
\bar{\cG}_n=\{ M_1\otimes \cdots \otimes M_n: M_j\in\{I,X,Y,Z\}\}
\end{align}
by ignoring the phase $\pm 1, \pm i$.
Note that $\cS$ can  be chosen to be a subgroup of $\bar{\cG}_n$.
Then the error analysis of the quantum code
is equivalent to  an classical additive block code in $\mF_{2^2}$~\cite{CRSS98}.
Thus the performance of the quantum code is completely characterized by the stabilizer group $\cS$.

In the scheme of EA quantum codes, maximally-entangled states $(\ket{00}+\ket{11})/{\sqrt{2}}$ are shared between the sender and receiver, which is an $+1$ eigenstate of $X\otimes X$ and $Z\otimes Z$.
(For more details about EA quantum codes, please refer to \cite{BDM06,BDM062}.)
Suppose  $\mathcal{S}\in \bar{\mathcal{G}}_{n+c}$ is a stabilizer group with $n-k+c$ independent generators
\begin{align*}
\cS=& \{U(I^{\otimes k}\otimes S^Z\otimes S^E )U^{\dag}\otimes S^E:  S^Z\in \{I,Z\}^{\otimes n-k-c},\\
& S^E\in\bar{\cG}_c \}.
\end{align*}
Then $\cS$ defines an $\left[\left[  n,k;c\right]  \right]$ EA stabilizer code, which is a $2^k$-dimensional subspace of $\mC^{2^n}\otimes \mC^{2^c}$.
Herein, we assume that before encoding, the first $k$ qubits  are information qubits,  the last $2c$ qubits are the maximally-entangled states, and the other $n-k-c$ ancillas  begin in $\ket{0}$.
It is assumed that the qubits held by the receiver before communication are error-free and thus
we neglect  operators on those qubits.  
 Then the simplified stabilizer group $\mathcal{S}'\subset \bar{\cG}_n$ is
\begin{align*}
\cS'=& \{U(I^{\otimes k}\otimes S^Z\otimes S^E )U^{\dag}:  S^Z\in \{I,Z\}^{\otimes n-k-c},\\
& S^E\in\bar{\cG}_c \}.
\end{align*}
 Now $\mathcal{S}'$ is no longer Abelian
and its centralizer group is
$$
C(\cS')= \{U(L\otimes S^Z\otimes I^{\otimes c})U^{\dag}:  S^Z\in \{I,Z\}^{\otimes n-k-c}, L\in \cG_k \}.
$$

\subsection{Quantum Convolutional Codes}

Like classical CCs,
a sequence of seed transformation encoders are applied
 in the case of quantum convolutional codes
as shown in Fig. \ref{fig:EAQ convolutional encoder}.
Note that entanglement (with corresponding operators $E_j$ in Fig.~\ref{fig:EAQ convolutional encoder}) is introduced to complete the duality notion of quantum stabilizer codes~\cite{LBW13}; that is, the dual code of a stabilizer code is an EA stabilizer code.
Please refer to \cite{PTO09,WHZ14} for more details about quantum convolutional  codes.

\begin{figure}
\centerline{
    \begin{tikzpicture}[scale=0.7][very thick]
    \fontsize{6pt}{1} 
    \tikzstyle{halfnode} = [draw,fill=white,shape= underline,minimum size=1.0em]
    \tikzstyle{checknode} = [draw,fill=blue!10,shape= rectangle,minimum height=6em, minimum width=2em]
    \tikzstyle{variablenode} = [draw,fill=white, shape=circle,minimum size=0.8em]
    \node (p1) at (2,2.75) {$P_{j-1}$} ;
    \node (w1) at (-2,2.75) {$M_{j-1}$} ;
    \node (l1) at (-2,2.25) {$L_{j-1}$} ;
    \node (s1) at (-2,1.75) {$S_{j-1}$} ;
    \node (e1) at (-2.1,1.25) {$E_{j-1}$} ;
    \node (w3) at (4,-0.25) {$M_{j+1}$} ;
    \node (p3) at (4,1.25) {$P_{j}$} ;
    \node (w2) at (1,1.5) {$M_{j}$} ;
    \node (l2) at (0,0.75) {$L_{j}$} ;
    \node (s2) at (0,0.25) {$S_{j}$} ;
    \node (e2) at (0,-0.25) {$E_{j}$} ;
    \node (label51) at (-1,3.5) {$\ddots$} ;
    \node (label52) at (3,-0.5) {$\ddots$} ;
    \node (label1) at (-1.25,0.15) {receiver} ;
    \node (label2) at (0.65,-1.35) {receiver} ;
    %
     \draw (w1)-- ++(1.65,0) (l1) -- ++(1.65,0) (s1) -- ++(1.65,0) (-1.5,1.25) -- ++(1.65,0)  ;
     \draw  (l2) -- ++(1.65,0) (s2) -- ++(1.65,0) (0.4,-0.25) -- (w3)  ;
     \draw (p1)-- ++(-1.65,0)   (p3)-- ++(-3.65,0) ;
     \draw  (-2,0.75)-- (-1.5,1.25) (-2,0.75)-- (-1.5,0.25) (-1.5,0.25) -- ++(0.45,0);
     \draw  (0,-0.75)-- (0.4,-0.25) (0,-0.75)-- (0.4,-1.25) (0.4,-1.25) -- ++(0.45,0);
    \node[checknode] (cn1) at (0,2) {$U_{(j-1)}$};
    \node[checknode] (cn3) at (2,0.5) {$U_{(j)}$};
   \end{tikzpicture}
 }

  \caption{
Circuit diagram of an EA-QCC encoder with a seed transformation $U$.
  }\label{fig:EAQ convolutional encoder}
  \vspace{-1.cm}
\end{figure}
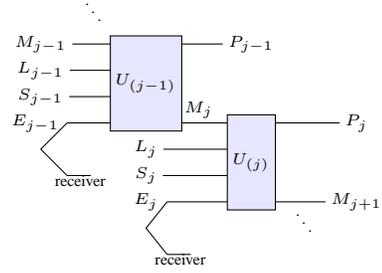

Suppose an   $(\!(n,k,c,m)\!)$ EA-QCC $\cC$
has $m$ memory qubits, $n$ output qubits,  and $k$ logical qubits  with the help of $c$ maximally-entangled states at each time step.
Suppose $\cC$ is defined by seed transformation operator $U_{(j)}$, which is an Clifford encoder on $\mC^{2^{n+m}}$.
 Let
 \begin{align*}
 \mathcal{I}^{M}=&\{1,\dots,m \},\\
 \mathcal{I}^{L}=&\{m+1,\dots,m+k\},\\
 \mathcal{I}^{A}=&\{m+k+1,\dots,m+n-c\},\\
 \mathcal{I}^{E}=&\{m+n-c+1,\dots,m+n\},\\
 \mathcal{I}^{M'}=&\{1,\dots,m\}, \\
 \mathcal{I}^P=&\{m+1,\dots,m+n\},
 \end{align*}
 which specify the locations of input memory qubits, logical qubits, ancilla qubits,  entangled qubits, and output memory qubits, and physical qubits, respectively.
  The seed transformation $U_{(j)}$ with the input parameters $(\mathcal{I}^M,\mathcal{I}^E,\mathcal{I}^A,\mathcal{I}^L)$ defines the dual constraint code $\mathcal{C}_{(j)}^{\perp}$  as  in Definition \ref{def:constraint code}.
\bd
The constraint code $\mathcal{C}_{(j)}$ of  $\cC$ at time $j$ is an $[[n+2m,k;c]]$ EA stabilizer code defined by simplified stabilizer group $\mathcal{S}_{(j)}$ with the following generators:
\begin{align*}
Z_i^M&\otimes g_i,\ X_i^M \otimes  h_i, &\mbox{ for $i\in \mathcal{I}^M$};\\
I^M&\otimes g_i, \ I^M \otimes  h_i, &\mbox{ for $i\in \mathcal{I}^E$};\\
I^M&\otimes g_i, &\mbox{ for $i\in \mathcal{I}^A$},
\end{align*}
where $Z_i^M, X_i^M, I^M \in \bar{\cG}_m$ and $g_i=U_{(j)}Z_i U^{\dag}_{(j)}$, $h_i=U_{(j)}Z_i U^{\dag}_{(j)}$ $\in \bar{\cG}_{n+m}$.
\ed

Recall that we consider stabilizer groups in $\bar{\cG}_n$, which is isomorphic to $\mF_{2^2}^n$.
Then the normal factor graph duality theorem \cite{Forney11,Forney01}\footnote{Here we employ the normal factor graph duality theorem, which is more general than our development of CCs in Lemma~\ref{def_dual_CCC}.}
suggests that the dual EA-QCC of $\cC$ is defined by the  constraint code
 $\mathcal{C}_{(j)}^{\perp}$, which is the $[[n+2m,c;k]]$ dual code of $\cC_{(j)}$ and has a simplified stabilizer group $\cS^{\perp}_{(j)}$ with the following generators:
\begin{align*}
Z_i^M&\otimes g_i,\ X_i^M \otimes  h_i, &\mbox{ for $i\in \mathcal{I}^M$};\\
I^M&\otimes g_i, \ I^M \otimes  h_i, &\mbox{ for $i\in \mathcal{I}^L$};\\
I^M&\otimes g_i, &\mbox{ for $i\in \mathcal{I}^A$}.
\end{align*}



An EA-QCC and its dual are uniquely defined up to a \emph{unitary row operator} $R$ that preserves  $\cS_{(j)}$ and $\cS_{(j)}^{\perp}$  \cite{LB10}.
 For example, $U_{(j)}R$ is a seed transformation that defines the same EA-QCC as $U_{(j)}$ does if for all $g\in\cS_{(j)}$ and $h\in\cS_{(j)}^\perp$, $RgR^\dagger\in \cS_{(j)}$ and $RhR^\dagger \in \cS_{(j)}^\perp$, respectively.

Let  $M_{(j)}, M_{(j+1)}\in\bar{\cG}_m$ be the memory operators at time  steps $j$ and $j+1$, respectively,
and let $\{L_{(j)} \in \bar{\cG}_k\}$ be the stream of logical operators.
The seed transformation $U_{(j)}$  produces a truncated\footnote{Note that the part of operators on memory qubits are discarded in this truncated stabilizer group, which is slightly different in the definition of the simplified stabilizer group of an EA stabilizer code.}
 stabilizer group $\overline{\cS}_{(j)}$ and a logical set $C(\overline{\cS}_{(j)})$ based on $M_j$ at each time step $j$.
 Consequently, the EA-QCC $\cC$ is the state space stabilized by the stabilizer group 
 $$\bigotimes_j \overline{\cS}_{(j)}= \{ \bigotimes_j g_{(j)}: g_{(j)}\in \overline{\cS}_{(j)} \}.$$
More precisely, we have
\begin{align*}
\overline{\cS}_{(j)}=& \{P_{(j)}\in \{I,X,Y,Z\}^{\otimes n}: P_{(j)}\otimes M_{(j+1)}
\\&= U_{(j)}(M_{(j)}\otimes I^{\otimes k} \otimes S^Z\otimes S^E )U_{(j)}^{\dag}\text{ for }  M_{(j)}\in \bar{\cG}_m, \\
&S^Z\in \{I,Z\}^{\otimes n-k-c}, S^E\in \bar{\cG}_c \},
\end{align*}
\begin{align*}
&C(\overline{\cS}_{(j)})=\\
& \{P_{(j)}\in \cG_n: P_{(j)}\otimes M_{(j+1)}=U(M_{(j)}\otimes L_{(j)}\otimes S^Z\otimes I^{\otimes c})U^{\dag} \\
&\text{ for } M_{(j)}\in \{I,X,Y,Z\}^{\otimes m}, S^Z\in \{I,Z\}^{\otimes n-k-c}, L_{(j)}\in \bar{\cG}_k \}.
\end{align*}
Note that an $(\!(n,k,m)\!)$ QCC is a special case of $c=0$.
%

Remark: One can define a polynomial check matrix $S(D)$ as in \cite{FGG07}
and show that $S(D)$ can be obtained from the seed transformation $U$ by finding an equation similar to (\ref{eq:G_from_ABCE}).

The weight generating function of a set $\cS\subset\mathcal{G}_n$ is  $$g_{\cS}(x,y)=\sum_{w=0}^n \nu_w x^{n-w}y^w,$$
where $\nu_w$ is the number of elements in $\cS$ of weight $w$. The WAMs of EA-QCCs are defined similarly to the classical case as follows.
\bd
The WAM $\Lambda_{\cC_{(j)}}(x,y)$ of an EA-QCC $\cC$ with  constraint code $\cC_{(j)}$ defined by a simplified stabilizer group $\cS_{(j)}$ is the matrix whose $(M_{(j)},M_{(j+1)})$ entry is
the weight generating function of the set of physical output operators $\{P_{(j)}\in C(\overline{\cS}_{(j)})\}$ when the input and output memory operators are $M_{(j)}$ and $M_{(j+1)}$, respectively. 
\ed

Consider $\bar{\mathcal{G}}_1=\{ I,X,Y,Z\}$.
By (\ref{eq:Fourier}), the matrix representation of the
Fourier transform  operator $\cF_{\bar{\cG}_1}$ in the ordered basis $\ket{I},\ket{X},\ket{Y},\ket{Z}$ is
\begin{align}
\cF_{\bar{\cG}_1}=
\begin{array}{c}
\left(\begin{array}{cccc} 1&1&1&1\\ 1&1&-1&-1 \\ 1&-1&1&-1 \\1&-1&-1&1 \end{array}\right), \label{eq:pauli fourier matrix}
\end{array}
\end{align}
where the symplectic inner product  (\ref{eq:Pauli inner product}) is used.
Then the  MacWilliams identity for EA-QCCs is as follows.
\bt \label{thm:MacWilliamsIdentity_EAQ}
Suppose the WAM of an $(\!(n,k,c,m)\!)$ EA-QCC $\cC$ is $\Lambda_{\cC_{(j)}}(y)$.
Then the WAM of its dual $\cC^{\perp}$ is
\begin{align*}
&\Lambda_{{\cC}^{\perp}_{(j)}}(y)=\frac{(1+3y)^n}{4^{m }4^{k}2^{n-k-c}}\cF_{\bar{\cG}_1}^{\otimes m}\Lambda_{\cC_{(j)}}\left(\frac{1-y}{1+3y}\right) \cF_{\bar{\cG}_1}^{\otimes m}.
\end{align*}
\et
\begin{proof}
Since the WAM $\Lambda_{\cC_{(j)}}(y)$ is counting the weight enumeration of the centralizer group $C(\cS_{(j)})$ of an $[[n+2m,k,c]]$ stabilizer code,
we have
\[
|C(\cS_{(j)})|=2^{2m}2^{2k}2^{n-k-c}.
\]
By Lemma~\ref{lemma:fourier transform},
\[\cF_{\bar{\cG}_m}=\cF_{\bar{\cG}_1}^{\otimes m}.
\]
The rest of the proof parallels that of Theorem \ref{thm:MacWilliamsIdentity_WAM} and  is omitted.
\end{proof}

The MacWilliams identity holds for the EA-QCCs, catastrophic or noncatastrophic, recursive or nonrecursive (see Refs \cite{PTO09,WHZ14} for these definitions).
For examples please refer to~\cite{LH14a}.

\section{Conclusion} \label{sec:conclusion}

With a different representation of the EWGF in our paper,
we provided a direct proof of the MacWilliams identity for the convolutional codes.
This method allows us to develop the MacWilliams identity for the IOWAMs of a CC and its dual with systematic encoders, which answered an open question in \cite{GLS09}. The input-output weight distributions are an essential part in the error analysis of iterative decoding, in particular for turbo codes. Our result could potentially lead to preliminary error analysis of both classical and quantum turbo codes.
Applications of the IOWAM for turbo codes can be found in \cite{RU08}, for example,
where they explicitly demonstrated the usefulness of IOWAMs for turbo codes.
By considering expected weight distributions of parallel and serially concatenated ensembles,  the error floor under MAP decoding can be predicted by input-output weight enumerators of minimal codewords in the asymptotic case.
In addition, one of us has developed an error analysis of linear block codes based on the MacWilliams identity for IOWGF in \cite{LKY04}. A similar error analysis for CCs and turbo codes could be potentially obtained, and is our ongoing investigation.

Mathematically speaking, various forms of weight enumerators can be seen as generalizations of one another, and each comes from a different {way of enumerating} the underlying abelian group, the $\mathbb{F}_q$ vector space.
Practically, these weight enumerators have rich applications in analyzing the error performance of a code. To elaborate, the usual Hamming weight enumerator can be used to provide an expression for the ``codeword error probability" of a code \cite{MS77}. On the other hand, to get down to the bit error rate (BER), one must require the input-output weight generating function (IOWGF) \cite{RU08}. Yet, when it comes to the non-binary channel of an $M$-ary modulation, the IOWGF is no longer sufficient for obtaining an expression for BER; it is the {complete weight generating function} being called for. Perhaps the most extreme case is when the channel is non-binary and time-varying, that is, the channel statistics varies at each transmission; then to deduce a BER expression in such a scenario we must have the {exact} weight generating function (EWGF). To sum up, the usual Hamming weight enumerator is only applicable to an error analysis for the simplest channel and offers very limited information about the structure of a code; on the contrary, the EWGF, which we used in the paper, is the most general one and is applicable to many other channels.

\section*{Acknowledgment}
CYL was supported by the Australian Research Council (ARC) under Grant DP120103776.  MH was supported by the UTS Chancellors postdoctoral research fellowship and UTS Early Career Researcher Grants Scheme during the early development of this work. MH is now supported by an ARC Future Fellowship under Grant FT140100574.

\end{document}